\documentclass[runningheads]{llncs}
\usepackage{fullpage}
\usepackage{amsmath}
\usepackage{algorithm,algpseudocode}
\usepackage{cleveref}
\usepackage{listings, fancyvrb, multicol}
\usepackage{xcolor}
\usepackage{graphicx}
\newcommand{\ignore}[1]{}
\newcommand{\myparagraph}[1]{\vspace{7pt}\noindent\textbf{#1}}

\newcommand{\privK}{K_{priv}}
\newcommand{\pubK}{K_{pub}}
\newcommand{\Log}{log}

\ignore{
\newtheorem{theorem}{Theorem}[section]
\newtheorem{lemma}[theorem]{Lemma}

\newtheorem{claim}[theorem]{Claim}

\newtheorem{definition}{Definition}

}

\newcommand{\correct}{correct}

\makeatletter
\lst@Key{countblanklines}{true}[t]%
{\lstKV@SetIf{#1}\lst@ifcountblanklines}

\lst@AddToHook{OnEmptyLine}{%
	\lst@ifnumberblanklines\else%
	\lst@ifcountblanklines\else%
	\advance\c@lstnumber-\@ne\relax%
	\fi%
	\fi}
\makeatother

\begin{document}
\title{Classifying Trusted Hardware via Unidirectional Communication}

\date{\today}

\author{Naama Ben-David\inst{1} \and Kartik Nayak\inst{2}}

\institute{VMware Research - \texttt{bendavidn@vmware.com} \and Duke University - \texttt{kartik@cs.duke.edu}}

\maketitle
\begin{abstract}
    It is well known that Byzantine fault tolerant (BFT) consensus cannot be solved in the classic asynchronous message passing model when one-third or more of the processes may be faulty. Since many modern applications require higher fault tolerance, this bound has been circumvented by introducing \emph{non-equivocation mechanisms} that prevent Byzantine processes from sending conflicting messages to other processes. The use of  \emph{trusted hardware} is a way to implement non-equivocation.
    
    Several different trusted hardware modules have been considered in the literature. In this paper, we study whether all trusted hardware modules are equivalent in the power that they provide to a system. We show that while they do all prevent equivocation, we can partition trusted hardware modules into two different power classes; those that employ shared memory primitives, and those that do not. We separate these classes using a new notion we call \emph{unidirectionality}, which describes a useful guarantee on the ability of processes to prevent network partitions. We show that shared-memory based hardware primitives provide unidirectionality, while others do not.
\end{abstract}

\section{Introduction}

Byzantine fault tolerance (BFT) is a fundamental problem in distributed computing, which has benefited from growing interest in recent years due to its application in blockchain technologies. BFT consensus allows a group of $n$ processes to commit on the same value even if up to $t$ of these processes behave arbitrarily. Depending on the problem formulation, this value may be proposed by a designated sender or all processes. It is well known that in practical distributed networks where processes communicate through asynchronous (or partially synchronous) message passing, $n \geq 3t+1$ processes are needed to tolerate $t$ Byzantine faults~\cite{dwork1988consensus}. Unfortunately, in many applications of BFT, tolerating only the failures of up to one-third of the network may not be enough. 

To circumvent this bound, at a high level we need stronger guarantees from the hardware that facilitates communication. For example, if we can guarantee \emph{synchronous} communication, 
i.e., any message sent by an honest process reaches its destination within $\Delta$ time where $\Delta$ is a known bounded message delay, then using a public key infrastructure, we can achieve higher fault tolerance~\cite{dolev1983authenticated}.
However, assuming synchrony often means that $\Delta$ needs to be very large, making such algorithms slow in practice.
Another approach for increasing fault tolerance augments the processes in the asynchronous (or partially synchronous) distributed network with some \emph{trusted hardware} which restricts the Byzantine actions a malicious process can perform~\cite{aguilera2019impact,bessani2008sharing,chun2007attested,levin2009trinc,correia2010asynchronous}.   Example trusted hardware modules include Trusted Platform Module (TPM), Intel Software Guard extensions (SGX)~\cite{costan2016intel}, Single-Writer Multi-Reader (SWMR) registers~\cite{aguilera2019impact,bouzid2016necessary,malkhi2003objects}, sticky bits~\cite{malkhi2003objects} and PEATS~\cite{bessani2008sharing}. Intuitively, all of these trusted hardware primitives provide a way to prevent \emph{equivocation}, disallowing a Byzantine process from sending conflicting messages to two different processes.

It has been shown that trusted hardware is a weaker assumption than synchrony; synchronous systems can solve a strictly larger set of tasks than systems augmented with trusted hardware while tolerating the failures of a minority of the processes~\cite{malkhi2003objects,bhadauria2018brief}.
Intuitively, the synchrony assumption is stronger since it allows us to detect equivocation by using a guaranteed delivery assumption. On the other hand, trusted hardware provides us with non-equivocation, but there is no assumption made on guaranteed delivery. 
However, while the difference between trusted hardware and synchrony is well understood, there has been little work on understanding the differences between trusted hardware options.

In this work, we are interested in the following question: \emph{Are all trusted hardware primitives equivalent, or are some provably stronger than others?}
%
While the different primitives are useful to improve fault tolerance, ensuring that they hold in practical systems requires significantly different approaches. Thus, answering this question can lead to a deeper understanding of the protocols involved and the properties they provide, and help determine which primitives hardware designers should invest in. 

\begin{figure}
    \centering
    \includegraphics[width=0.7\textwidth]{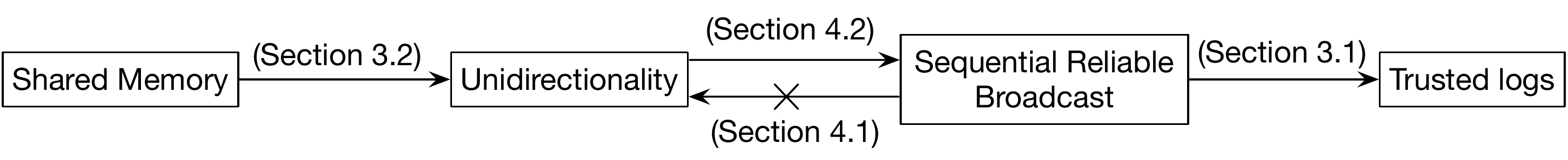}
    \caption{\textbf{Summary of results.} A $\rightarrow$ B indicates A can implement B}
    \label{fig:my_label}
\end{figure}
\ignore{
\myparagraph{Separating synchrony and trusted hardware.} Under Byzantine consensus, processes may equivocate and send different values to different processes. We observe that the synchrony assumption allows us to detect these equivocating values since it provides us with  \emph{bidirectionality}, i.e., for any pair of honest correct processes that send a message to each other in a round, both processes receive the other process's messages at the end of the round. This eventually provides us with the desirable property of non-equivocation, enabling tolerance of a minority corruption. On the other hand, different trusted hardware modules directly disallow parties from equivocating without providing bidirectionality.~\kartik{cite clement?} 
We formalize the separation by \kartik{complete} (Section~\kartik{refer}).
}

We answer the above question by showing a separation between two distinct classes of trusted hardware. Intuitively, these classes correspond to \emph{shared memory primitives} and those that do not provide shared memory. We prove this separation by defining a new property, called \emph{unidirectionality}, which can be achieved by all  shared memory hardware primitives (e.g., SWMR registers, sticky bits, PEATS), but not by message passing primitives (e.g., trusted counters, trusted logs). Finally, to show the separation is strict, we provide an implementation of message passing trusted primitives from a system that provides unidirectionality.

\section{Preliminaries}\label{sec:prelims}
We consider an asynchronous distributed system with $n$ processes, up to $f$ of which may be Byzantine. A Byzantine process may behave arbitrarily. If a process is not Byzantine, we say that it is \emph{\correct}. We assume processes have access to \emph{unforgeable transferable signatures}. In addition to signatures, we assume that processes can communicate via one of the communication methods described in the next subsection, all of which are known to implement some form of \emph{non-equivocation}. 

There have been several studies of non-equivocation mechanisms~\cite{clement2012limited,jaffe2012price,madsen2020subject}, and hardware that provides it is becoming increasingly practical~\cite{levin2009trinc,aguilera2019impact,aguilera2020microsecond}. It is known that a non-equivocation mechanism, which prevents Byzantine processes from sending conflicting messages to different processes, can increase the fault tolerance of the system. More specifically, a system with non-equivocation and transferable signatures can tolerate the corruptions of any minority of the processes when solving \emph{weak Byzantine agreement}. However, \emph{strong Byzantine agreement}, which differs from its weak counterpart by considering only the inputs of \correct{} processes to be valid, cannot be solved in such systems with the same fault tolerance. Indeed, any asynchronous or partially synchronous system can only solve strong agreement when there are at least $n > 3f$ processes~\cite{malkhi2003objects}.  This is true even in the crash failure model. On the other hand, synchronous systems with transferable signatures \emph{can} solve strong agreement with $n >2f$~\cite{dolev1983authenticated}. 


\subsection{Non-Equivocation Mechanisms in the Literature}\label{sec:noneqLit}

\myparagraph{Mechanisms with trusted logs.} Attested append-only memory (A2M)~\cite{chun2007attested} provides a trusted log on which any process can \emph{append} a value, and receive an \emph{attestation} with the index of this value in this log. Past log entries cannot be modified.  Levin et al.~\cite{levin2009trinc} simplify the assumptions required by A2M by introducing a \emph{trusted incrementer (TrInc)}. To prevent equivocation, a sending process must increment the TrInc counter and attach the resulting sequence number to its message. The TrInc guarantees that no two messages can have the same sequence number attached to them. From the perspective of providing non-equivocation guarantees, Intel SGX and ARM TrustZone are similar to A2M and TrInc. Though, in addition, they allow for more expressive computations. 




\myparagraph{Shared memory with ACLs.} 
In shared memory, to tolerate any Byzantine failures at all, we must assume that the Byzantine processes cannot write to all memory locations; otherwise, they can completely overwrite the memory, thereby preventing communication among \correct{} processes. To this end, shared memory primitives have been associated with \emph{access control lists (ACLs)}. These lists specify, for each object $O$ and operation $op$, which processes can execute $op$ on $O$. A special case of this is \emph{single-writer multi-reader (SWMR) registers}, which allows every process to invoke a \emph{read} operation on every register, and each register has an \emph{owner} process, which is the only process allowed to \emph{write} on this register. Studies of SWMR registers in the context of Byzantine fault tolerance appeared in~\cite{aguilera2019impact,bouzid2016necessary,malkhi2003objects}. Other shared memory primitives that have been studied in this context include \emph{sticky bits}~\cite{malkhi2003objects}, which are registers whose values cannot be changed after the first write, and \emph{Policy-enforced augmented tuple spaces (PEATS)}, which allow for inserting, removing, and reading typed entries from data structures called tuples. PEATS control access to object operations not just through static ACLs, but through \emph{policies} that can take into account the state of the object at the time of the attempted operation~\cite{bessani2008sharing}. It has been shown that SWMR registers can solve weak Byzantine agreement with $n\geq 2f+1$~\cite{aguilera2019impact}, and that none of these shared memory primitives can solve strong Byzantine agreement with $n\leq 3f$~\cite{malkhi2003objects}. Thus, it is clear that such shared memory primitives are stronger than asynchronous message passing and weaker than synchronous message passing model.

\section{Categorizing Non-Equivocation Mechanisms}

In this section, we define two useful notions: sequenced reliable broadcast (SRB) and unidirectionality. We show that trusted log primitives are weaker than (SRB), and that shared memory primitives are stronger than unidirectionality. In the next section, we complete the separation by showing that unidirectionality is strictly stronger than SRB.

\subsection{Trusted Logs are Weaker than Sequenced Reliable Broadcast}
The first primitive we consider is called \emph{sequenced reliable broadcast}. Similar definitions have appeared in the literature~\cite{aguilera2019impact,guerraoui2006introduction}. Intuitively, this primitive is similar to reliable broadcast, but enforces sequence numbers on the messages, which must be broadcast and delivered in this order. 

\begin{definition}[Sequenced Reliable Broadcast]
	In \emph{sequenced reliable broadcast}, there is a designated process $p$ called the \emph{sender} that can broadcast any number of messages, each with a unique sequence number, such that the following conditions hold:
	\begin{enumerate}
		\item \label{prop:SRBweak-termination} If $p$ is \correct{}, then every \correct{} process eventually delivers every message that $p$ broadcasts.
		\item \label{prop:SRBstrong-termination} If some \correct{} process $q$ delivers message $m$ with sequence number $k$ from $p$, then eventually every \correct{} process delivers $m$ with sequence number $k$ from $p$.
		\item \label{prop:SRBseq} If some \correct{} process $q$ delivers a message with sequence number $k$ from $p$ at time $t$, then $q$ delivered messages with all sequence numbers $1\leq k' < k$ from $p$ before time $t$.
		\item \label{prop:SRBintegrity} If some \correct{} process delivers a message $m$ from $p$, then $p$ broadcast $m$ at some earlier point in time.
	\end{enumerate}
\end{definition}

To formally compare the power of SRB to trusted hardware primitives, we begin by defining the functionality that these primitives provide. In their paper, Levin et al.~\cite{levin2009trinc} show that TrInc can implement the interface of attested append-only memory (A2M)~\cite{chun2007attested}. Therefore, to show that both primitives are weaker than SRB, we simply have to show that SRB can implement TrInc. Algorithm~\ref{alg:hw-interface} presents a formal interface for TrInc. 
Intuitively, TrInc provides each process with access to its own \emph{Trinket}, which it can use to get \emph{attestations} of messages it would like to send. A process must provide its message and a sequence number to the \emph{Attest} function in order to get an attestation. The process can then send an attestation of its message, which contains the message itself. Processes receiving an attestation can verify that it was produced by a valid Trinket, but using the CheckAttestation function with the id of the process that sent the message. A Trinket does not produce a new valid attestation for a  sequence number that has already been used.
This interface is a simplification of the way TrInc was presented in~\cite{levin2009trinc}, keeping the parts of its that affect its theoretical power, and omitting those parts that were put in place for improving its practicality in real systems.


\begin{figure}
\begin{lstlisting}[postbreak=\mbox{\textcolor{red}{$\hookrightarrow$}\space}, breaklines=true, keywords={}]
Process p can invoke the following functions on its Trinket @$T_p$@

    attestation Attest(seq-num c, message m)
        Returns a valid attestation @$a$@ attesting to (prev, c, m), if c is higher than any seq-num used for an attestation on this Trinket so far. prev is the sequence number of the last attested value.
        Returns null otherwise
        
    bool CheckAttestation(attestation @$a$@, id q)
        Return true if @$a$@ is a valid attestation that was previously output by Trinket @$T_q$@. Return false otherwise
\end{lstlisting}
    \caption{TrInc Interface}
    \label{alg:hw-interface}
\end{figure}

\begin{theorem}
Sequenced Reliable Broadcast can implement the interface specified in Algorithm~\ref{alg:hw-interface}.
\end{theorem}

\begin{proof}
   We present an implementation of the TrInc functionality using sequenced reliable broadcast. 
 
    \begin{lstlisting}
for each process p, 
initialize: k = 0, for each process q, C[q] = 0

attestation Attest(seq-num c, message m) {
    Broadcast(k, (c,m)); // k is the broadcast sequence number
    return (k, (c,m));}

bool CheckAttestation(attestation a, id q){
    upon delivering a message (k, c, m) from q
        if C[q] < c { @\label{line:checkseq}@
            store (k, (c,m));
            C[q] = c;   }
    if (I've stored a message (k,(c,m)) == a from q) {
        return true; }
    else{
        return false;}  }
    \end{lstlisting}
    
We now show that the above implementation satisfies the properties of the TrInc interface. In particular, we show that (1) $CheckAttestation(a,q)$ would eventually return true if process $q$ correctly invoked a $T_q$.Attest instance that returned $a$, and (2) $CheckAttestation(a,q)$ returns false if $a$ was not correctly attested by $T_q$.


For the first property, recall that by property~\ref{prop:SRBweak-termination}, if $q$ correctly invoked broadcast for value $(k,c,m)$, then eventually every correct process will deliver $(k,c,m)$ from $q$. Furthermore, $T_q$.Attest(c,m) will return $(k,c,m)$ in this case.  Therefore, eventually, some invocation of $CheckAttestation(a,q)$ by each correct process $p$ will happen after the process already delivered $(k,c,m)$. Recall that a correct attestation always uses a sequence number $c$ that is larger than all previous ones. Therefore, if $q$ correctly attested its message delivered by $p$, the check on Line~\ref{line:checkseq} passes, and $p$ stores this message. So the $CheckAttestation(a,q)$ call by $p$ will return true.

Secondly, by property~\ref{prop:SRBintegrity} of SRB, if a correct process $p$ delivered an attestation $a$ from $q$, $q$ must have broadcast it. Therefore, by definition, there was an Attest by $q$ that returned $a$. A $CheckAttestation(a,q)$ by $p$ returns false if $p$ did not deliver $a$ from $q$. 
\end{proof}

 \ignore{  
   \begin{lstlisting}
function CreateLog(){
    Initialize new SRB instance with sender id=C
    Return C    }

function Append(id, x){
    Have process @$id$@ broadcast $x$   }

function Lookup(id, s, z){
    Check whether message with sequence number @$s$@ has been delivered from sender @$id$@.
    If so, return the value delivered.  }

function End(id, z){
    Return the most recent value delivered from sender @$id$@.  }
   \end{lstlisting}
   
   Note that the Append function is the only place in which broadcast is called, and that each broadcast instance is associated with exactly one log. Furthermore, by property~\ref{prop:SRBintegrity}, no value can be returned by Lookup or End that wasn't broadcast by the sender of the corresponding instance of SRB. Therefore, Lookup and End cannot return a value that was not Appended to the appropriate log.
   
   \Naama{This abstraction actually sounds like shared memory. If it is, then it is actually stronger than SRB. The issue is with linearizing when the $c_{id}$ is incremented. In the implementation given, I don't think there's a way to linearize. Is the intention of the hw functionality that it be linearizable? Also, some weirdness in that any process can call lookup and end in this implementation, but createLog and append are called by some different entity.}
}

\subsection{Shared Memory is Stronger than Unidirectionality}
Next, we define a new notion, called \emph{unidirectional communication}. Intuitively, a system with unidirectional communication is partially immune to network partitions, as it can implement rounds in which there is at least some communication between every pair of \correct{} processes. 

\begin{definition}[Unidirectional communication]
	A system provides \emph{unidirectional communication} if it can implement \emph{rounds} with the following property:
	
	For any pair of correct processes $p$ and $q$, if both $p$ and $q$ send a message to each other in round $r$, then either $p$ receives $q$'s message before the beginning of $p$'s next round or $q$ receives $p$'s message before the beginning of $q$'s next round.
\end{definition}

We now show that trusted hardware that is based on shared memory can implement unidirectional rounds. That is, this hardware is at least as strong as unidirectionality.

Consider a shared memory setting with $n$ objects $o_1 \ldots o_n$, such that object $o_i$ can be modified by process $p_i$ and read by all processes. We note that all shared memory objects that have some modifying operation and some read operation, along with access control lists (ACL)~\cite{malkhi2003objects} can provide this setting. This includes SWMR registers, PEATS, and all objects considered in~\cite{malkhi2003objects}. We show this setting can achieve unidirectional communication.

\begin{claim}
    Consider a shared memory system $S$ in which for each process $p_i$, there is some object $o_i$ such that $p_i$ is the only process that can modify $o_i$, and all processes can read $o_i$. Unidirectional communication can be implemented in $S$.
\end{claim}
\begin{proof}
	We present an implementation of a unidirectional round using $n$ objects $o_1 \ldots o_n$, such that $o_i$ allows only process $p_i$ to modify it, and all processes to read it. This protocol was first introduced in~\cite{aguilera2019impact} to implement a weak notion of broadcast using SWMR registers.
	
	\begin{lstlisting}
	In round r, process @$p_i$@ executes the following:
    	To send message m, @$p_i$@ appends (r, m) in object @$o_i$@
    	@$p_i$@ reads objects @$o_1 \ldots o_n$@
	    @$p_i$@ is said to receive a round r message m' from process @$p_j$@ if @$p_i$@ reads M from @$o_j$@ such that
	        there is some message (r, m') in M
	\end{lstlisting}
	
	Assume by contradiction that the above implementation of rounds does not satisfy unidirectionality. That is, there exist two correct processes $p_i$ and $p_j$ that both send a message in round $r$, but neither receive the other's message. Assume without loss of generality that $p_i$ wrote its round $r$ message in $o_i$ before $p_j$ did so in $o_j$. Then, since $p_j$ must write its message before reading $o_i$, $p_j$ must see $p_i$'s round $r$ message when it reads $o_i$.
	Contradiction.
\end{proof}

\ignore{
\begin{definition}
A \emph{round} $r$ for a  process $p$ is if $p$ sends a message to some set $S$ of processes, then blocks until some event happens.
\end{definition}

\begin{definition}
We say a communication round $r$ has \emph{unidirectional} communication for process $p$ if, if $p$ send a round $r$ message to a \correct{} process $p'$, then by the end of $p$'s round $r$, either $p$ has received a round $r$ message from $p'$, or $p$ learns that $p'$ will hear $p$'s message by the end of $p'$'s round $r$. 
\end{definition}

\begin{definition}[Unidirectional communication between a pair of processes]
We say that a pair of correct processes $p$ and $q$ have unidirectional communication if when $p$ and $q$ send a message to each other, then either $p$ receives the message sent by $q$, and $q$, if it does not receive $p$'s message, receives a notification that $p$ has received its message, or vice versa. 
\end{definition}
}

\section{Separation Between Unidirectionality and Sequenced Reliable Broadcast}

We now show that sequenced reliable broadcast is strictly weaker than unidirectionality except in some corner cases. Intuitively, the result holds because of the inability of reliable broadcast to break through a network partition between two \correct{} processes. This result therefore holds even for stronger variants of non-equivocation, as long as the only guarantee they provide is \emph{eventual} delivery.

\subsection{SRB Cannot Implement Unidirectionality}
\begin{claim}
	Sequenced reliable broadcast cannot implement unidirectionality in a system with $n > 2f$ and $f>1$, under asynchrony and in the absence of additional assumptions.
\end{claim}

\begin{proof}
	Assume by contradiction that there exists a protocol $P$ that uses sequenced reliable broadcast and implements a unidirectional round in a system with $n > 2f$ and $f>1$. Partition the processes into three sets, $Q$, $C_1$, and $C_2$, where $|Q| = n-f$, $|C_1| = 1$, and $|C_2| = f-1$. 
	Consider the following scenarios.
	
	\paragraph{Scenario~1.} The process $p \in C_1$ is faulty, and all the others are \correct. $p$ crashes at the beginning of the execution and never sends any messages, and all messages from $C_2$ to processes in $Q$ are arbitrarily delayed. All other messages are received immediately. Then processes in $Q$ must eventually start the next round since, from the perspective of $Q$, $C_1$ and $C_2$ could both have been faulty, in which case the number of faults is $\leq f$. This satisfies the problem constraint. Similarly, $C_2$ must eventually start the next round, since they are receiving all messages sent by correct processes, which are $> n-f$. 
	
	Observe that in this scenario, a process in $C_2$ starts the next round without receiving a message from $C_1$.
	
	\paragraph{Scenario~2.} The processes in $C_2$ are faulty, and the rest are \correct. No process in $C_2$ sends any messages, and all messages from $C_1$ to $Q$ are arbitrarily delayed. All other messages are received immediately. Using the same reasoning as in Scenario~1, processes in $Q$ and $C_1$ must eventually start the next round.
	
	Observe that in this scenario, a process in $C_1$ starts the next round without receiving a message from any party in $C_2$. 
	
	\paragraph{Scenario~3.} No process is faulty, but all messages out of $C_1$ and $C_2$ to other sets are arbitrarily delayed and all other messages arrive immediately. This scenario is indistinguishable to processes in $Q$ from both of the other scenarios. Furthermore, it is indistinguishable to $C_1$ from Scenario~2, and indistinguishable to $C_2$ from Scenario~1.
	Therefore, any pair of processes $p \in C_1$ and $p' \in C_2$ does not receive each other's message in this round despite both being correct and sending messages. This contradicts the unidirectionality property.
\end{proof}

In Appendix~\ref{sec:corner}, we show that in the corner case where $n\geq 3$ and $f=1$, sequenced reliable broadcast can in fact implement unidirectionality.

\subsection{Unidirectionality Can Implement SRB}
Conversely, we show that unidirectional communication \emph{can} implement sequenced reliable broadcast. An algorithm solving sequenced reliable broadcast using unidirectional communication is presented in Figure~\ref{alg:uni-srb}. This result relies on the construction presented by Aguilera et al.~\cite{aguilera2019impact}, in which they show an algorithm for a broadcast primitive that is equivalent to sequenced reliable broadcast, implemented from SWMR registers. In this paper, we show that the construction can be rewritten to assume only unidirectional rounds instead. 

Intuitively, the algorithm of Aguilera et al.~\cite{aguilera2019impact} relies on the construction of \emph{proofs} that enough processes have received a given value from the sender. Processes read the sender's register, and copy over the value that they see into their own slot, appending their signatures to it. They then scan all the registers until they see at least $t+1$ copies of the same value as their own, and no other value. If they reach this state, then they create an \emph{L1 proof}; they copy over all of the signed copies of the sender's message into their own slot, and append a new signature to it. This process is repeated another time; processes now scan the array until they see enough L1 proofs, all attesting to the same value. At this point, they again copy the L1 proofs into their own slot, constructing an \emph{L2 proof}. Once an L2 proof is constructed, a process may deliver the message.

The crux of the correctness argument for this algorithm relies on the fact that no two correct processes $q,q'$ can construct contradicting L1 proofs: to do so, $q$ and $q'$ would have had to copy over different values from the sender, and since they scan all registers before creating their respective L1 proofs, at least one of them (w.l.o.g., $q'$), must have seen the other's contradictory value. This prevents $q'$ from constructing its own proof. 
Since no two correct processes create contradicting L1 proofs, no two contradicting L2 proofs can be created (whether by correct or Byzantine processes). Therefore, once an L2 proof is generated, anyone seeing it can safely adopt this value.

Adapting this algorithm to unidirectional communication is simple; we replace all `write' operations with `send to all', and replace all `read' operations with receiving a message.
Interestingly, the property of SWMR that prevents correct processes from generating contradicting L1 proofs is captured by unidirectionality; a correct processes $q$ starts a unidirectional round when it forwards the value it received from the sender. By the time this round ends, $q$ must have received the value sent by every correct process that did not receive $q$'s message. In particular, this means that if the sender sent conflicting values to different correct processes, at least one of them will be aware of this conflict and fail to produce an L1 proof.

\renewcommand{\figurename}{Algorithm}
\begin{figure}
	\caption{Sequenced Reliable Broadcast using Unidirectional Rounds}
	\begin{lstlisting}[columns=fullflexible,breaklines=true]
Code @\texttt{for}@ process p
	next_p; //next index to deliver from sender. Initially, next_p=1
	state; //@$\in$@{WaitForSender,WaitForL1Proof,WaitForL2Proof}. Initially, state=WaitForSender
	my_seq; //most recent sequence number used to broadcast a message, if p is the sender
	
	void broadcast (m){
	    my_seq += 1;
		send sign((my_seq,m)) to all; @\label{line:broadcastWrite}@   }
	
	bool maybeDeliver() {
	    k = next_p;
	    val = checkL2Proof(k);@\label{line:readL2proof}@
		if (val != null) {
			deliver(k, val, q);@\label{line:deliver}@
			next_p += 1;
			state = WaitForSender;
			return true;	}
	    return false;
	}
	
	void try_deliver(q) {
		if (state == WaitForSender) {
		    upon receiving message val=(j,m) from q
		        if (maybeDeliver(val)){return;}
			    if (!sValid(q, val) || j!=k) {	return;	}
			    Send sign(val) to all;@\label{line:copyVal}@
			    state = WaitForL1Proof; }
	
		if (state == WaitForL1Proof) {
			checkedVals = @$\emptyset$@;
			do {
			    upon receiving message v from process r @\label{line:readMessage}@
				    if (maybeDeliver(val)){return;}
				    if (validateValue(v,val,k,q)) {checkedVals.add((r,v));} 
			} until (unidirectional round is finshed and size(checkedVals) @$\geq$@ t+1)
	
			l1prf = sign(checkedVals);
			send l1prf to all;@\label{line:writel1prf}@
			state = WaitForL2Proof; }   }
	
		if (state == WaitForL2Proof){
			checkedL1Prfs = @$\emptyset$@;
			do{
			    upon receiving message prf from r @\label{line:readL1proof}@
				    if (maybeDeliver(val)){return;}
				    if (validateL1Prf(prf,val,k,q)) {checkedL1Prfs.add((r,prf));}
			} until (unidirectional round finished and size(checkedL1Prfs) @$\geq$@ t+1)
	
		l2prf = sign(checkedL1Prfs);
		send l2prf to all; }   } }@\label{line:writeL2proof1}@
	\end{lstlisting}
\label{alg:uni-srb}
\end{figure}

\begin{figure}
	\caption{Helper functions for Sequenced Reliable Broadcast algorithm}
\begin{lstlisting}[firstnumber=52]
value checkL2proof(proof,k) {
    if proof contains at least one sequence number and at least one value{
        j = first sequence number in proof;
        if (j@$\neq$@k) {return null;}@\label{line:seqnum}@
		val = first value in proof;
		if (validateL2Prf(proof,val,k,p)) {
			Send proof to all; @\label{line:writeL2proof2}@
			return val; }   }   }
	return null;    }
	
bool validateValue(v, val, k, q){
	if (v == val && sValid(q,v) && key == k){
		return true; }
	return false;	}
	
bool validateL1Prf(proof, val, k, q){
	if (size(proof) @$\geq$@ t+1) {
		for each (i,(v, s)) in proof {
			if (!validateValue(v,val,k,q) || !sValid(i, (v,s)){
				return false;	} 	}	}
	return true;	}

bool validateL1Prf(proof, val, k, q){
	if (size(proof) @$\geq$@ t+1 && @$\forall$@ 
		for each (i,(l1prf,s)) in proof {
			if (!validateL1Proof(l1prf,val,k,q) || !sValid(i, (l1prf, s))){
				return false;}	} 	}
	return true;	}
\end{lstlisting}
\label{alg:uni-srb-help}
\end{figure}

We arrive at the following claim.

\begin{claim}\label{lem:SMnon-eq}
	Sequenced Reliable Broadcast can be solved using unidirectional communication with $n\geq 2t+1$ processes.
\end{claim}

We prove the lemma by showing that Algorithm~\ref{alg:uni-srb} correctly implements sequenced reliable broadcast. 
We do so by showing that the algorithm satisfies each of the required properties.

\begin{lemma}\label{lem:deliverseq}
If some \correct{} process $q$ delivers a message with sequence number $k$ from $p$ at time $t$, then $q$ delivered messages with all sequence numbers $1\leq k' < k$ before time $t$.
\end{lemma}
\begin{proof}
Processes deliver a message $(k,m)$ on Line~\ref{line:deliver}, which is only executed if checkL2Proof(p,k) returns a non-null value, which a correct process calls with $k$ equal to the next sequence number to be delivered from the sender. This sequence number starts at $1$ and is only ever incremented upon delivery of a value from $p$. Note that by Line~\ref{line:seqnum} in checkL2Proof(), checkL2Proof returns null if the sequence number does not match. The rest of the proof can be completed by a straightforward induction.
\end{proof}
	
\begin{lemma}
If the sender $p$ is correct, then every correct process eventually delivers every message that $p$ broadcasts.
\end{lemma}

\begin{proof}
Let $p$ be a correct sender. We assume by contradiction that there exists some message $(k,m)$ the $p$ broadcasts, but that some correct process $q$ never delivers. Furthermore, assume without loss of generality that $k$ is the smallest sequence number for which $p$ broadcasts a message that some correct process never delivers. That is, all correct processes must eventually deliver all messages $(k',m')$ from $p$, for $k' < k$. Thus, all correct processes must eventually increment $last[p]$ to $k$.

We consider two cases, depending on whether or not some process eventually sends an L2 proof for some $(k,m')$ message from $p$.

First consider the case where no process ever sends an L2 proof of any value $(k,m')$ from $p$.
Since $p$ is correct, upon broadcasting $(k,m)$, $p$ must send a signed copy of $(k,m)$ (line~\ref{line:broadcastWrite}). Since $p$ is correct, it sends $(k,m)$ to all processes, and never sends any other message with that sequence number to any process. So, every correct process will eventually receive  $(k,m)$ in line~\ref{line:readMessage}, sign and send it to all others, and change their state to WaitForL1Proof. 

Furthermore, since $p$ is correct and we assume signatures are unforgeable, no process $q$ can send any other valid value $(k',m')\neq(k,m)$ to any correct process $r$ and have that value validated by $r$. Thus, eventually each correct process will add at least $t+1$ copies of $(k,m)$ to its checkedVals, sign and send an L1proof consisting of these values, and change their state to WaitForL2Proof. 

Therefore, all correct processes will eventually receive at least $t+1$ valid L1 Proofs for $(k,m)$ in line~\ref{line:readL1prf} and construct and send valid L2 proofs for $(k,m)$. This contradicts the assumption that no L2 proof is ever sent. 

In the case where there is some L2 proof, by the argument above, the only value it can prove is $(k,m)$. 
Therefore, all correct processes will receive at least one valid L2 proof and deliver. This contradicts our assumption that $q$ is correct but does not deliver $(k,m)$ from $p$.
\end{proof}
	
\begin{lemma}
If some correct process $q$ delivers message $m$ with sequence number $k$ from $p$, then eventually every correct process delivers $m$ with sequence number $k$ from $p$.
\end{lemma}
\begin{proof}
Let $q$ be a correct process that delivers $(k,m)$, and let $q'$ be another correct process. Assume by contradiction that $q'$ never delivers $(k,m)$ from $p$. 

We consider two cases.

\textsc{Case 1.} $q'$ eventually delivers some other message $(k,m')$ from $p$, where $m\neq m'$.

Since $q$ and $q'$ are correct, they must have received valid L2 proofs at line~\ref{line:readL2proof} before delivering $(k,m)$ and $(k,m')$ respectively. Let $\mathcal{Q}$ and $\mathcal{Q'}$ be those valid proofs for $(k,m)$ and $(k,m')$ respectively. $\mathcal{Q}$ (resp. $\mathcal{Q'}$) consists of at least $t+1$ valid L1 proofs; therefore, at least one of those proofs was created by some correct process $r$ (resp. $r'$). Since $r$ (resp. $r'$) is correct, it must have sent $(k,m)$ (resp. $(k,m')$) to all processes in line~\ref{line:copyVal}. Recall that both $r$ and $r'$ communicate through unidirectional rounds, and wait until their round ends before compiling and sending their L1Proofs. Since both $r$ and $r'$ are correct, at least one of them must have received the other's message before compiling its L1 proof.
Assume without loss of generality that $r'$ received $r$'s message. Since $r'$ is correct, it cannot have then compiled an L1 proof for $(k,m')$. We have reached a contradiction.

\textsc{Case 2.} $q'$ never delivers any message from $p$ with sequence number $k$.

By Lemma~\ref{lem:deliverseq}, if $q$ delivers $(k,m)$ from $p$, then for all $i<k$ there exists $m_i$ such that $q$ delivered $(i,m_i)$ from $p$ before delivering $(k,m)$.

Assume without loss of generality that $k$ is the smallest key for which $q'$ does not deliver any message from $p$. Thus, $q'$ must have delivered $(i,m_i')$ from $p$ for all $i<k$; thus, $q'$ must have incremented $last[p]$ to $k$. Since $q'$ never delivers any message from $p$ for sequence number $k$, $q'$'s $last[p]$ will never increase past $k$.

Since $q$ delivers $(k,m)$ from $p$, then $q$  must have sent a valid L2 proof $\mathcal{P}$ of $(k,m)$ in line~\ref{line:writeL2proof1} or~\ref{line:writeL2proof2}. Thus, $q'$ will eventually receive this message. Since $q'$'s $last[p]$ eventually reaches $k$ and never increases past $k$, $q'$ will eventually call checkL2Proof with sequence number $k$, and checkL2Proof will eventually return a non-null value, causing $q'$ to deliver a value for sequence number $k$. We have reached a contradiction.
\end{proof}

\begin{lemma}
If some \correct{} process delivers a message $m$ from $p$, then $p$ broadcast $m$ at some earlier point in time.
\end{lemma}

\begin{proof}
We show that if a correct process $q$ delivers $(k,m)$ from a correct sender $p$, then $p$ broadcast $(k,m)$.
Correct processes only deliver values for which a valid L2 proof exists (lines~\ref{line:readL2proof}---\ref{line:deliver}). Therefore, $q$ must have received a valid L2 proof $\mathcal{Q}$ for $(k,m)$. $\mathcal{Q}$ consists of at least $t+1$ L1 proofs for $(k,m)$ and each L1 proof consists of at least $t+1$ matching copies of $(k,m)$, signed by $p$. Since $p$ is correct and we assume signatures are unforgeable, $p$ must have broadcast $(k,m)$ (otherwise $p$ would not have attached its signature to $(k,m)$). 
\end{proof}

\bibliographystyle{plain}
\bibliography{main}

\appendix
\section{SRB Can Implement Unidirectionality When $n\geq 3$ and $f=1$}\label{sec:corner}

For the corner case, we show that sequenced reliable broadcast can implement unidirectionality when $n\geq 3$ and $f=1$. In fact, only reliable broadcast is needed for this result to hold. 

\begin{claim}
	Reliable broadcast can implement unidirectionality in a system with $f=1$ and $n\geq 3$.
\end{claim}
\begin{proof}
	Consider the following protocol for creating a unidirectional round.
	
	\begin{lstlisting}[keywords={}]
	Protocol for process p with input v:
	Phase 1: send (v,@$\sigma_p$@) to all //where @$\sigma_p$@ is @$p$@'s unforgeable signature for @$v$@
	    wait to receive phase 1 messages with valid signatures from n-1 distinct processes
	Phase 2: forward all messages received to all
	    wait to receive phase 2 messages from n-1 distinct processes, 
	    such that each message is of the from [(@$v_1$@, @$\sigma_1$@), ... (@$v_m$@, @$\sigma_m$@)]  
	    where m @$\geq$@ 2 and all signatures are valid and from distinct processes.
	\end{lstlisting}
	
	We claim that at the end of phase~2 of a \correct{} process $p$, the unidirectional property holds for $p$ with every other \correct{} process $p'$ for the input values.
	
	
	Consider two correct processes $p$ and $p'$ executing the above protocol in a system with reliable broadcast, and let $Q$ be  the set containing the rest of the processes in the system. If $p$ receives $p'$'s message (or vice versa) directly in either phase, then the unidirectional property already holds. 
	
	Thus, assume that neither $p$ nor $p'$ receives the other's message directly in either phase.
	Note that all processes in $Q$ must receive at least one of $p$ or $p'$'s messages in phase~1, without loss of generality assume they receive $p$'s value. Furthermore, both $p$ and $p'$ receive all of $Q$'s phase~2 messages. Since valid phase~2 messages must contain $n-1$ values from phase~1 and are unforgeable, $Q$'s phase~2  message must contain $p$'s phase~1 message, which $p'$ now receives.
	%
\end{proof}
\end{document}